\documentclass[11pt]{amsart}
\usepackage{graphicx}
\usepackage{tikz}
\usetikzlibrary{matrix}
\usepackage{multicol}
\usepackage{amssymb,amsthm,amsaddr}
\usepackage{bbm}
\usepackage{graphbox}
\usepackage{upgreek}
\usepackage{tcolorbox}
\usepackage{tabto}
\usepackage{braket}
\usepackage[export]{adjustbox}
\usepackage{nicematrix}
\usepackage{appendix}
\usepackage{soul}
\usepackage{comment}

\usepackage{mathtools} 
\usepackage{mathrsfs} 
\usepackage{stmaryrd} 

\usepackage[margin=1in]{geometry}
\usepackage[pagebackref, colorlinks = true, linkcolor = blue, urlcolor  = blue, citecolor = red]{hyperref}
\usepackage{xcolor}
\definecolor{darkgreen}{rgb}{0.09, 0.45, 0.27}
\definecolor{darkred}{rgb}{0.55, 0.0, 0.0}

\renewcommand{\epsilon}{\varepsilon}
\renewcommand{\phi}{\varphi}
\newcommand{\C}{\mathbb{C}}

\newcommand{\F}{\mathbb{F}}

\newcommand{\lnorm}{\lvert \!\vert}
\newcommand{\rnorm}{\vert \! \rvert}
\DeclareMathOperator{\Tr}{Tr}

\newtheorem{theorem}{Theorem}[section]
\newtheorem{definition}[theorem]{Definition}
\newtheorem{proposition}[theorem]{Proposition}
\newtheorem{corollary}[theorem]{Corollary}
\newtheorem{lemma}[theorem]{Lemma}
\newtheorem{remark}[theorem]{Remark}

\newcommand{\ketbra}[2]{| #1 \rangle \langle #2 |}
\newcommand{\bbraket}[2]{\langle #1 \vert #2\rangle}
\newcommand{\injnorm}[1]{\lnorm #1 \rnorm_{\mathrm{inj}}}

\newcommand{\rvline}{\hspace*{-\arraycolsep}\vline\hspace*{-\arraycolsep}}

\title{The Injective norm of CSS quantum error-correcting codes}
\author{Stephane Dartois}
\address{\'Ecole Polytechnique, Institut Polytechnique de Paris, Centre de Mathématiques Laurent Schwartz, 91120 Palaiseau, France}
\email{stephane.dartois@polytechnique.edu}
\author{Gilles Zémor}
\address{Institut de Mathématiques de Bordeaux, UMR 5251, 351 Cours de la Libération, 33400 Talence, France}
\address{Institut Universitaire de France}
\email{zemor@math.u-bordeaux.fr}

\begin{document}

\maketitle

\begin{abstract}
   In this paper, we compute the injective norm — \textit{a.k.a.} geometric entanglement — of standard basis states of CSS quantum error-correcting codes. The injective norm of a quantum state is a measure of genuine multipartite entanglement. Computing this measure is generically NP-hard. However, it has been exactly computed in condensed matter theory — notably in the context of topological phases — for the Kitaev code and its extensions, in works by Orús and collaborators. We extend these results to all CSS codes and thereby obtain the injective norm for a nontrivial, infinite family of quantum states. In doing so, we uncover an interesting connection to matroid theory and Edmonds' intersection theorem.
\end{abstract}

\section{Introduction}
\subsection{Context}
We are interested in the multipartite entanglement of the standard basis states of arbitrary CSS quantum error-correcting codes~\cite{Calderbank1996,Steane1996}. Such families of quantum states are of great importance in quantum information and computation, and some of them also appear frequently in condensed matter theory, indeed, LDPC codes can often be seen as phases of matter \cite{kitaev2010topological,haah2013lattice, yoshida2015topological, yin2024low, placke2024topological}. We show that the exact value of the geometric entanglement for these states is entirely dictated by the dimensions of the shortened codes of the classical $X$ code. Our work can be seen as a broad generalization of results obtained in the physics literature~\cite{orus2011topological,orus2014geometric}, with the key distinction that our arguments do not rely on any specific geometric or locality structure, and therefore apply to all CSS codes.

Among the motivations for this work lies the fact that the structure of many-body or multipartite entanglement is only poorly understood in spite of its importance for many applications. 

One of the well-known measures of multipartite entanglement is the injective norm, whose definition we recall in equation \eqref{eq:inj_norm_def} below. The goal of this work is to benchmark the injective norm against a family of interesting quantum states provided by quantum error correcting codes. Computing the injective norm of a quantum state is an NP-hard problem \cite{hillar2013most}, therefore finding families of states for which it can be explicitly computed is an interesting problem by and for itself. In fact, there are a number of works in the literature which study the behavior of the injective norm of families of simple enough quantum states \cite{wei2003geometric,zhu2010additivity} or focus on finding maximally entangled multipartite states with respect to the injective norm \cite{aulbach2010maximally, steinberg2024finding}. 

A series of strong results were recently obtained concerning the typical behavior of the injective norm of random tensors and random quantum states, both in mathematics~\cite{dartois2024injective,bandeira2024geometric,boedihardjo2024injective,bates2025balanced,stojnic2025ground} and from the viewpoint of physicists~\cite{sasakura2024signed}. The methods used in these studies broadly come from landscape complexity, statistical physics, and geometric functional analysis. The motivations range from questions about (classical) locally decodable codes, quantum information, statistical physics of spin glasses, and data analysis. Moreover, the injective norm has been extensively studied numerically in~\cite{fitter2022estimating} for different families of random and deterministic states.

In this context, our work can both be seen as contributing to a program aimed at better understanding tensor norms as well as to a program dedicated to the study of multipartite entanglement measures. The latter have been intensively studied in recent years for random states, with a particular focus on random tensor network states, due to their role as excellent toy models for Ryu-Takayanagi-like formulas. The studied measures of multipartite entanglement aim to generalize Rényi entropies~\cite{penington2023fun,akers2024reflected}, and are constructed out of polynomial local unitary invariants of tensors~\cite{collins2023tensor,jing2012classes,cui2017local}, mirroring how Rényi entropies arise from polynomial local unitary invariants of density matrices. In the multipartite case, the profusion of such invariants obscures their operational meaning and complicates the choice of a canonical family with good properties. By focusing on the injective norm, we step outside the polynomial framework and analyze a non-polynomial quantity.

Let $n$ be a positive integer, let $(\mathcal{H}_i)_{i=1}^n$ be a collection of Hilbert spaces (of possibly different dimensions $d_i$). Given a quantum state $\lvert \psi\rangle \in \bigotimes_{i=1}^n \mathcal{H}_i$, the injective norm of such a state is defined by 
\begin{equation}\label{eq:inj_norm_def}
    \injnorm{\lvert \psi \rangle}:=\max_{\lvert \phi_i\rangle \in \mathcal{H}_i, \langle \phi_i\vert \phi_i\rangle=1} \lvert \langle \psi \vert \phi_1\ldots \phi_n\rangle\rvert.
\end{equation}
One often considers minus the logarithm of the above quantity, called the geometric entanglement,
$$E(\lvert \psi \rangle):=-\log_2 \injnorm{\lvert \psi \rangle}^2.$$
Yet another equivalent quantity, which is just as natural from a geometric point of view, is the distance of $\lvert \psi \rangle$ to the set of separable states, namely
$$d_{\textrm{SEP}}(\lvert \psi \rangle)=\sqrt{2(1-\injnorm{\lvert \psi \rangle})}.$$
This last quantity is sometimes called the Groverian \cite{jung2008reduced}. We see in particular that 
the larger the injective norm of a state is, the closer this state is to a separable state.

In the $n=2$ bipartite case, where $\ket{\psi}\in \mathcal{H}_1\otimes \mathcal{H}_2$, the geometric entanglement is just the $\infty$-Rényi entropy of the density matrix $\rho=\Tr_2(\ketbra{\psi}{\psi})$.

\subsection{Results}
Following standard coding theory practice, we call a subvector space of $\F_2^n$ a {\em classical linear code} over $n$ bits, where $\F_2$ denotes the finite field on two elements.
We shall require very little coding theory background, but for a textbook treatment of classical error-correcting codes, see \cite{macwilliams1977theory}, while \cite{nielsen2000quantum} treats the quantum case.

Let $C$ be a $k$-dimensional linear code over $\F_2^n$. We are interested in the associated quantum state that we denote by $\ket{C}$ defined as,
\begin{equation}
    \label{eq:ket(C)}
\ket{C}=\frac1{\sqrt{\lvert C\rvert}}\sum_{y\in C}\ket{y} \in (\mathbb{C}^2)^{\otimes n}.
\end{equation}
Let $A\subset [n]$ be a subset of coordinate positions. We define the {\em punctured code} on $A$ to be the code on $|A|$ bits consisting of all codewords of $C$ restricted to coordinates of $A$.
Define $j(C)$ to be the smallest integer $j$ such that there exists a partition of $[n]=A\sqcup B$  for which the punctured code on $A$ is of dimension $k$ while the punctured code on $B=[n]\setminus A$ is of dimension $k-j$.

\begin{theorem}[Main theorem]\label{thm:Main}
Let $C\subset \F_2^n$ be a linear code of dimension $k$. Then the injective norm of $\ket{C}$ is given by 
\[
\injnorm{\ket{C}}=2^{-\frac12(k-j(C))}.
\]
\end{theorem}
Theorem~\ref{thm:Main} will be derived by proving matching upper and lower bounds on the injective norm of the state $\ket{C}$. The upper bound reads $\injnorm{\ket{C}}\le 2^{-\frac12(k-j(C))}$ while the lower bound is expressed in terms of an additional quantity $\delta(C)$. One obtains that $\injnorm{\ket{C}}\ge 2^{-\frac12(k-\delta(C))}$, with $\delta(C):=\max_{C_0}2k_0-\ell(C_0)$ 
where the maximum is over shortened codes $C_0$ of $C$ (see definition \ref{def:shortened_codes}), $k_0$ is the dimension of $C_0$ and $\ell(C_0)$ its length. 
Proving the upper and lower bounds is the object of Section~\ref{sec:up-low-bound}.

The most technical part of the proof is to show that the upper and lower bounds match, namely that $j(C)=\delta(C)$. This last equality is the content of Theorem \ref{thm:pre_main} and 
Section~\ref{sec:main_equality} is devoted to its proof. The core argument involves using matroid theory and notably Edmonds' intersection theorem, an abstract version of the max flow-min cut principle.

\begin{remark}
The family of states of the form  
\begin{equation}
    \label{eq:ket{C}}
\ket{C}=\frac1{\sqrt{|C|}}\sum_{y\in C}\ket{y}
\end{equation}
captures a number of states that we are commonly interested in. An observation is that the celebrated GHZ states,  
$$
\ket{\text{GHZ}_n}=\frac1{\sqrt{2}}\bigl(\ket{\,\underbrace{0\ldots 0}_{n \text{ times}}}+\ket{\,\underbrace{1\ldots 1}_{n \text{ times}}}\bigr),
$$ 
belong to this family: indeed, they correspond to the case when $C$ is the well-known repetition code. On the other hand, some canonical examples fall outside this class. For instance, the $W$ state,  
$$
\ket{W}=\tfrac1{\sqrt{3}}(\ket{100}+\ket{010}+\ket{001}),
$$  
cannot be written in the form~\eqref{eq:ket{C}}, since the words $100,\,010,\,001$ do not form a linear subspace of $\F_2^3$ 
\end{remark}

Moreover, we will make the case that Theorem~\ref{thm:Main} extends and gives a formula for the injective norm of all basis states of any quantum CSS code. To be specific, let $\mathcal{H}=(\C^2)^{\otimes n}$ be the total Hilbert space of $n$ qubits.
Let $C_1\subseteq \F_2^{n}, C_2\subseteq \F_2^{n}$ be two classical binary linear codes satisfying $C_2\subset C_1$. 
We recall that $(C_1,C_2)$ defines a {\em quantum CSS code} \cite{Calderbank1996,Steane1996} that consists of the subspace
of $\mathcal{H}$ equal to the linear span of all states
\begin{equation}\label{eq:basis_states}
    \ket{z}:=\frac1{\sqrt{\lvert C_2\rvert}}\sum_{y\in z}\ket{y}
\end{equation}
where $z$ ranges over all cosets $z\in C_1/C_2$.
A simple consequence of Proposition \ref{prop:injnorm_const_LU_classes} below and Theorem \ref{thm:Main} is that, 
\begin{corollary}
Let $C_2$ be a code of dimension $k$ in $\F_2^n$. For all $x\in C_1/C_2$, $\injnorm{\ket{x}}=2^{-\frac12(k-j(C_2))}$. 
\end{corollary}
This means that standard basis states of a CSS code all have the same injective norm.

{\it Relation to previous results:}

The geometric entanglement of the state $\ket{C}$ is 
$$E(\ket{C})=k-j(C).$$ In the case of the Kitaev toric code, we have a code $C$ of length $n$ and the dimension $n-1$,
and the work \cite{orus2011topological} of Or\'us and Wei shows that the geometric entanglement of basis states is $n-1$, that is $j(C)$ vanishes for the Kitaev code. Therefore, the Kitaev code produces basis states which are maximally entangled for the geometric entanglement among basis states of CSS codes. The work \cite{orus2011topological} went further and computed the topological geometric entanglement by grouping qubits in larger and larger clusters. Their clustering method (reminiscent of block spin renormalization) is highly dependent on a locality property, which is not obviously generalized in the case of general CSS codes. We therefore postpone the study of the relevant generalization of topological geometric entanglement for CSS codes to further work. 

{\it Computational complexity for $j(C)$:}

Regarding the problem of computing the injective norm of tensors, one noticeable fact is that, when specialised to the basis states of a CSS quantum error-correcting code, the problem becomes discrete,  whereas the initial optimization problem on a product of spheres is of a continuous nature.

A natural question is whether this discretization reduced the difficulty of computing the injective norm. It turns out that our matroid formulation of the problem implies the existence of an efficient (polynomial in the number of qubits $n$) greedy augmenting path algorithm that computes $j(C)$: this is a consequence of the work of Edmonds \cite{edmonds1971matroids}.
\section{Acknowledgments}
S.D. is grateful to the Institut de Mathématiques de Bordeaux for their hospitality, where part of this work was carried out. The work of S.D. was partly supported by the ANR grants ANR-25-CE40-1380 and ANR-25-CE40-5672. G.Z. was supported by
Plan France 2030 through the project NISQ2LSQ, ANR-22-PETQ-0006.
\section{Preliminaries}
\subsection{General facts on injective norm}
Let $\mathcal{H}$ be a Hilbert space of the general form $\mathcal{H}=\bigotimes_{i=1}^n \mathcal{H}_i$.
Let us start by pointing out that the injective norm \eqref{eq:inj_norm_def} of a state clearly satisfies {\em local unitary invariance}, namely:
\begin{proposition}
\label{prop:invariance}
Let $U_i\in U(\mathcal{H}_i)$, then $\injnorm{\ket{\psi}}=\injnorm{(U_1\otimes\ldots \otimes U_n)\ket{\psi}}$, where each $U_i$ is a unitary operator acting separately on each Hilbert space $\mathcal{H}_i$.
\end{proposition}

Next we mention a generic method for deriving upper bounds on the injective norm.

Pick a basis $\{\ket{\alpha_{m_i}^{(i)}}_{i=1}^{d_i}\}$ for each $\mathcal{H}_i$. Using these bases, $\ket{\psi}=\sum_{m_i\in [\![1,d_i]\!]}t_{m_1,\ldots,m_n}\lvert \alpha_{m_1}^{(1)}\ldots \alpha_{m_n}^{(n)}\rangle$. The elements $T=(t_{m_1,\ldots, m_n})\in \bigotimes_{i=1}^{n}\C^{d_i}$ form a tensor.
\begin{definition} \label{def:flattening}
    Let $A\sqcup B=[\![1,n]\!]$ be a non-trivial bipartition of $[n]:=[\![1,n]\!]$. Then a {\em flattening} of a tensor $T$ is the matrix $M(T)=(M(T)_{a,b})$ such that $a,b$ are multi-indices taking values in $a\in \prod_{i_A\in A}[d_{i_A}], b\in \prod_{i_B\in B}[d_{i_B}]$, and the elements $M(T)_{a,b}:=t_{a\sqcup b}$, where $a\sqcup b$ is the multi-index in $\prod_{i=1}^n[d_{i}]$ whose elements $i_j, j\in A$ produce $a$ and elements $i_j, j\in B$ produce $b$.
\end{definition}
The bipartition above induces a bipartition of the total Hilbert space $\mathcal{H}=\bigotimes_{i=1}^n\mathcal{H}_i$ as $\mathcal{H}=\mathcal{H}_A\otimes \mathcal{H}_B$ with $\mathcal{H}_A:=\bigotimes_{i_A\in A} \mathcal{H}_{i_A}$ and $\mathcal{H}_B:=\bigotimes_{i_B\in B} \mathcal{H}_{i_B}$. A flattening is just the reinterpretation of a multipartite quantum state in $\bigotimes_{i=1}^n\mathcal{H}_i$ as a bipartite state in $\mathcal{H}_A\otimes \mathcal{H}_B$.

We shall use:
\begin{lemma}\label{lem:inj<opnorm}
    Let $T$ be a tensor, and $M(T)$ a flattening associated to a bipartition $A,B$ as above. Then 
    \begin{equation}
        \injnorm{T}\le \lnorm M(T)\rnorm_{op}.
    \end{equation}
\end{lemma}
Note that the operator norm of $M(T)$ is its largest singular value. This largest singular value can be accessed as the square root of the largest eigenvalue of $\rho_A=M(T)M(T)^*$. It is interesting to remark that $\rho_A$ is the partial trace over $B$ of the density matrix $\rho_{AB}=\ket{T}\bra{T}$ associated to $T$.

In the rest of the paper, we shall focus on the case when $d_i=2$ for every $i$, so that every component Hilbert space $\mathcal{H}_i$ is isomorphic to $\mathbb{C}^2$. The ambient Hilbert space will hereafter be equal to $\mathcal{H}=\otimes_{i=1}^n\mathcal{H}_i$.
Following standard practice we write the canonical basis of $\mathbb{C}^2$ as $\ket{0},\ket{1}$ and elementary product states as $\ket{x}$ for $x\in\F_2^n$. We shall also use the notation $\ket{+}=\frac{1}{\sqrt{2}}(\ket{0}+\ket{1})$ and $\ket{-}=\frac{1}{\sqrt{2}}(\ket{0}-\ket{1})$.

\subsection{Coding Theory}
A binary linear code of length $n$ is an $\F_2$-subvector space of the ambient space $\F_2^n$. It will be useful to have coordinates of a binary vector not necessarily indexed by consecutive integers $1,2,\ldots n$, therefore we will sometimes identify the ambient space $\F_2^n$ with $\F_2^E$ where $E$ is a finite set of cardinality $n$. When a code $C$ is defined in the ambient space $\F_2^E$, we will refer to the {\em length} of $C$ as $\ell(C)=|E|$. A {\em generator matrix} for the code $C$ is a matrix $\mathbf{G}$ whose rows form a basis of $C$ as an $\F_2$-vector space.
The {\em rate} $R(C)$ of a code $C$ is defined as the ratio of its dimension relative to its length: $R(C)=\dim C/\ell(C)$.

\begin{definition}
\label{def:shortened_codes}
Let $C\subset\F_2^n$ be a linear code of length $n$. Let $A\subset [n]$. The {\em punctured code} of $C$ on $A$ is defined as the image of $C$ by the map:
\begin{eqnarray*}
\F_2^n & \rightarrow & \F_2^A\\
(x_1,\ldots ,x_n) & \mapsto & (x_i)_{i\in A}
\end{eqnarray*}
The {\em shortened code} of $C$ on $A$ (or supported by $A$) is defined to be the kernel $C_0$ of the map:
\begin{eqnarray*}
C & \rightarrow & \F_2^{\bar{A}}\\
(x_1,\ldots ,x_n) & \mapsto & (x_i)_{i\in \bar{A}}
\end{eqnarray*}
where $\bar{A}=[n]\setminus A$: the corresponding set $A$ is referred to as the {\em support} of the shortened code.
The code $C_0$ may also be thought of as code in the ambient space $\F_2^A$, so that the length $\ell(C_0)$ of the shortened code will mean $\ell(C_0)=|A|$ and its rate $R(C_0)$ will mean $R(C_0)=\dim C_0/|A|$.
\end{definition}

Let $C$ be a code in $\F_2^n$ and let $x\in\F_2^n$ be a vector in the ambient space. Extending the definition \eqref{eq:ket(C)} of the state $\ket{C}$ to coset states \eqref{eq:basis_states}, namely
\[
\ket{x+C} = \frac{1}{\sqrt{|C|}}\sum_{y\in C}\ket{x+y},
\]
we obtain:
\begin{proposition}
\label{prop:coset}
    For every $x\in\F_2^n$ we have $\injnorm{\ket{x+C}}=\injnorm{\ket{C}}$.
\end{proposition}
\begin{proof}
Let $U_i$ be the unitary operator on the component Hilbert space $\mathcal{H}_i$ defined by $U_i\ket{0}=\ket{x_i}$ and $U_i\ket{1}=\ket{1+x_i}$. We clearly have
\[
(U_1\otimes\ldots \otimes U_n)\ket{C}=\ket{x+C}.
\]
The result therefore follows from Proposition~\ref{prop:invariance}.
\end{proof}

Recall that two classical codes $C_1,C_2$ such that $C_2\subseteq C_1$ define a quantum CSS code $CSS(C_1,C_2)$ that is a subspace in $\mathcal{H}$, whose basis states are defined to be the states $\ket{z}$ for $z\in C_1/C_2$. 
Proposition~\ref{prop:coset} implies:

\begin{proposition}\label{prop:injnorm_const_LU_classes}
    Let $\mathcal{Q}=\text{CSS}(C_1,C_2)$ be the quantum CSS code associated to two classical linear codes $C_1,C_2, \, C_2\subseteq C_1$. The injective norm $\injnorm{\cdot}$ is a constant function on the set of basis states for $\mathcal{Q}$.
\end{proposition}

\section{Upper and lower bounds}\label{sec:up-low-bound}
\noindent{\bf Lower bound.}
We start with the lower bound. 
\begin{lemma}[Lower bound]
\label{lem:lower_bound}
    Let $C\subseteq \F_2^n$ a code of dimension $k$. Then,
    \begin{equation}
        \injnorm{\ket{C}}\ge 2^{-\frac12(k-\delta(C))},
    \end{equation}
    with
    \begin{equation}
        \delta(C):=\max_{C_0} \left(2\dim C_0-\ell(C_0)\right)=\max_{C_0} \ell(C_0)(2R(C_0)-1)
    \end{equation}
    where the maximum is over all shortened codes $C_0$ of $C$, $\ell(C_0)$ is the length of $C_0$, and $R(C_0)$ is its rate, according to Definition~\ref{def:shortened_codes}.
\end{lemma}
\begin{proof}
Given a partition $A\sqcup B=[n]$ we compute the scalar product
$$\bbraket{+_A 0_B}{C}\le \injnorm{\ket{C}},$$
where by $\ket{+_A 0_{B}}$ we mean $\bigotimes_{i\in A}\ket{+}\bigotimes_{i\in B}\ket{0}$. 
Let $C_0$ be the shortened code of $C$ supported by $A$. We have:
\begin{align*}
    \bbraket{+_A 0_{B}}{C} &= \frac{1}{\sqrt{|C|}}\sum_{c\in C}\bbraket{+_A 0_{B}}{c}\\
                                 &= \frac{1}{\sqrt{|C|}}\sum_{c\in C_0}\bbraket{+_A 0_{B}}{c}\\
                                 &= \frac{1}{\sqrt{|C|}}\sum_{c\in C_0}\frac{1}{2^{|A|/2}}\sum_{x_A\in\F_2^A}\bbraket{x_A 0_{B}}{c}\\
                                 &= \frac{1}{\sqrt{|C|}}\frac{1}{2^{|A|/2}}|C_0|\\
                                 &=2^{-\frac12(k-2k_0+\ell(C_0))}
\end{align*}
where $k=\dim C$, $k_0=\dim C_0$, and $\ell(C_0)=|A|$. Optimizing over all shortened codes $C_0$ of $C$ proves the claim.
\end{proof}

\begin{remark}
$\delta(C)$ measures how well the dimension of $C$ is distributed over subspaces. In fact, if there exists a shortened code $C_0$ whose rate is strictly larger than $1/2$ then $\delta(C)>0$. In this case $C$ has the property that a large number of its  codewords (namely $2^{k_0}$ of them) has support on a restricted set of positions $A$. 

If there is no such shortened code, $\delta(C)=0$ (the maximum is attained for $A=\emptyset$, so that $k_0=0, \ \ell(C_0)=0$), and the lower bound on the injective norm only depends on the dimension of $k$. In this case the codeword supports are more evenly spread inside $C$.
\end{remark}

\noindent {\bf Upper bound.} 
Let $C\subset \mathbb{F}_2^n$ be a linear code of length $n$ and dimension $k$. We call $j(C)$ the smallest integer $0\le j \le k$ such that there exists a partition $[n]=A\sqcup B$ for which the code $C$ punctured on $A$ (see Definition~\ref{def:shortened_codes}) is of dimension $k$ and punctured on $B$ is of dimension $k-j$. In other words, if we let $\mathbf{G}$ be a generator matrix of the code, the sub-matrix $\mathbf{G}_A$ consisting of columns indexed by $A$ must be of rank $k$ and the sub-matrix  $\mathbf{G}_B$ must be of rank $k-j$. 

\begin{lemma}\label{lem:upper_bound}
Let $C\subset \mathbb{F}_2^n$ be a linear code of dimension $k$ with $j(C)=j$ then $\injnorm{\ket{C}}\le 2^{-\frac12(k-j)}.$
\end{lemma}
The proof relies on the decomposition of the coordinate set $[n]$ into an information set \cite{prange1962use} for $C$ and its complement. An information set is a subset of positions such that the codeword coordinates at those positions uniquely identify the codeword. Given a codeword we call prefix the subword of positions corresponding to the chosen information set of the code, while we call suffix the subword made of the remaining positions. Rephrasing the above, for any given prefix $c_{\textrm{pre}}$ of a codeword $c$ of $C$, there is a unique associated suffix, allowing one to reconstruct the full codeword given one prefix.
\begin{proof}
    We recall $\ket{C}:=\frac1{\sqrt{\lvert C\rvert}}\sum_{c\in C}\ket{c}\in (\mathbb{C}^2)^{\otimes n}$. Let $\mathbf{G}$ be the generator matrix of the code $C$. We assume, without loss of generality, that it is in standard form $$\mathbf{G}=\begin{pmatrix} \begin{matrix}
     1 &\ldots & 0\\
     \vdots & \ddots &\vdots \\
     0& \ldots&1
    \end{matrix}&\rvline&
    \begin{matrix}
&&\\
&\mathbf{A}&\\
&&
\end{matrix}\end{pmatrix}$$
meaning that the left sub-matrix is the $k\times k$ identity matrix, and the right submatrix is some $k\times (n-k)$ matrix $\mathbf{A}$ with $\textrm{rank}(\mathbf{A})=k-j$. 
The first $k$ positions thus form an information set, allowing us to split the words into prefix and suffix. This splitting induces a splitting of the Hilbert space into the Hilbert space of prefixes (of dimension $2^k$) and the Hilbert space of suffixes (of dimension $2^{n-k}$) so that we see $\ket{C}\in (\mathbb{C}^2)^{\otimes n}\simeq(\mathbb{C}^2)^{\otimes k}\otimes (\mathbb{C}^2)^{\otimes n-k}$. To this splitting is associated a flattening of $\ket{C}$ according to the partition $A=[k], B=[n]\backslash[k]$ (see definition \ref{def:flattening}). Denoting this flattening $M(\ket{C})$ we have according to lemma \ref{lem:inj<opnorm}, 
$$\injnorm{\ket{C}}\le \lnorm M(\ket{C})\rnorm_{\textrm{op}}.$$
The matrix $M(\ket{C})\in \textrm{Mat}_{2^k\times 2^{n-k}}(\mathbb{C})$ is the matrix whose lines are indexed by elements of $\mathbb{F}_2^k$, which is also the set of prefixes of all codewords of $C$, while columns are indexed by elements of $\mathbb{F}_2^{n-k}$. 
Let $\mathbf{a}\in \mathbb{F}_2^k$ and $\mathbf{b}\in \mathbb{F}_2^{n-k}$. There is a one at position $\mathbf{a}, \mathbf{b}$ of  $\sqrt{\lvert C\rvert }M(\ket{C})$ if the prefix $\mathbf{a}$ has for suffix $\mathbf{b}$ in $C$. Every other element is zero. The square of the operator norm of $M(\ket{C})$ is the largest eigenvalue of $\rho_{\textrm{suffix}}:=M(\ket{C})^*M(\ket{C}).$ 
We have for $\mathbf{b}, \mathbf{b}'\in \mathbb{F}_2^{n-k}$
$$\lvert C\rvert(\rho_{\textrm{suffix}})_{\mathbf{b},\mathbf{b}'}=\lvert C\rvert\sum_{\mathbf{a}\in \mathbb{F}_2^k}(M(\ket{C})^*)_{\mathbf{b},\mathbf{a}}(M(\ket{C}))_{\mathbf{a},\mathbf{b}'}=N(\mathbf{b})\delta_{\mathbf{b},\mathbf{b}'},$$
where $N(\mathbf{b})$ is the number of prefixes having $\mathbf{b}$ as suffix. 
Indeed, the element $\mathbf{b},\mathbf{b}'$ of $\lvert C\rvert\rho_{\textrm{suffix}}$ is the scalar product (over $\mathbb{C}$) of the line $\mathbf{b}$ of $\sqrt{\lvert C\rvert}M(\ket{C})^*$ with the column $\mathbf{b}'$ of $\sqrt{\lvert C\rvert}M(\ket{C})$ which is easily seen to count the number of prefixes having $\mathbf{b},\mathbf{b}'$ as suffixes. But thanks to the information set property that a prefix has a unique suffix, we must have $\mathbf{b}=\mathbf{b}'$, otherwise the scalar product vanishes. Therefore the largest eigenvalue of $\rho_{\textrm{suffix}}$ is the largest value of $N(\mathbf{b})$ (divided by $\lvert C \rvert$). 
The number of prefixes having the same suffix $\mathbf{b}$ counts the different linear combinations of lines of $\mathbf{A}$ leading to the same bitstring $\mathbf{b}$ or equivalently, (at the cost of adding $\mathbf{b}$ to those combinations), the number of different vanishing linear combinations of lines of $\mathbf{A}$, that is $N(\mathbf{b})=\lvert\textrm{ker}_{\mathbb{F}_2}  \mathbf{A} \rvert=2^{j}$. In particular, if $\mathbf{b}$ is indeed a suffix of a codeword in $C$, then $N(\mathbf{b})$ does not depend on $\mathbf{b}$. We conclude that 
\[
\lnorm M(\ket{C})\rnorm_{\textrm{op}}=\frac{2^{\frac12j}}{\sqrt{\lvert C\rvert}}=2^{-\frac12(k-j)}.
\qedhere
\]
\end{proof}

\section{Matching upper and lower bounds}\label{sec:main_equality}
Let $C\subset\F_2^n$ be a code such that $j(C)=0$. Then Lemmas~\ref{lem:upper_bound} and~\ref{lem:lower_bound} together imply that $\delta(C)=0$ so that the upper and lower bounds match. 
To prove Theorem~\ref{thm:Main} it will therefore be enough to prove:
\begin{theorem}\label{thm:pre_main}
    Let $C\subset \mathbb{F}_2^n$ be a linear code, such that $j=j(C)\ge 1$, then there exists $C_0$ a shortened code of $C$ of dimension $k_0>0$ and length $2k_0-j$. 
\end{theorem}
Indeed, if such a shortened code $C_0$ exists, then by definition of $\delta(C)$ we must have $\delta(C)\geq 2k_0 - (2k_0-j) =j$,
but then Lemmas~\ref{lem:upper_bound} and~\ref{lem:lower_bound} imply that the corresponding upper and lower bounds are equal.

We shall prove Theorem~\ref{thm:pre_main} using arguments from matroid theory. To this aim, we introduce the relevant background and state one of the cornerstones of the theory and essential element of the proof of theorem \ref{thm:pre_main}, namely Edmonds' intersection theorem.

\subsection{Elements of matroid theory}
Matroids are combinatorial objects that are defined to abstract and generalize the concept of linear independence of vectors in linear algebra. They can also be seen as a way to generalize concepts from graph theory. For a gentle introduction to matroids and some of their applications, we refer to the notes \cite{matroidsnotes} and the associated book \cite{oxley2006matroid}. For our purpose we need only a few concepts of matroid theory, that is the definition of a matroid, its bases, its dual, and the rank functions over a matroid. In particular, this allows us to state Edmonds' intersection theorem that we use in the next section to prove Theorem~\ref{thm:pre_main}.

One way to define a matroid over a finite set $X$ is through the data of independent subsets of $X$. 
\begin{definition}
Let $X$ be a finite set, a finite matroid $\mathcal{M}(X,\mathcal{I})$ on $X$ is a couple $(X,\mathcal{I})$ where  $\mathcal{I}\subseteq 2^X$ is a subset of the powerset of $X$ satisfying the following constraints
\begin{enumerate}
    \item $\emptyset\in \mathcal{I}$
    \item if $J\in \mathcal{I}$ and $I\subset J$ then $I\in \mathcal{I}$
    \item if $I,J\in \mathcal{I}$ such that $\lvert I\rvert<\lvert J\rvert$, there exists $x\in J\backslash I$ such that $I\cup\{x\}\in \mathcal{I}$.
\end{enumerate}
The set $\mathcal{I}$ is called the \textit{set of independents}, and its elements are the \textit{independent sets}.
\end{definition}
The archetypical example of a matroid is the following.
    Let $V$ be a dimension $N$ vector space over a field $K$ endowed with a specific generating set of vectors $X=\{e_1,\ldots,e_R\}, R>N$. Let $\mathcal{I}$ be the set of linearly independent subsets of $X$. This endows $(V,X)$ with the structure of a matroid. Additionally, if one thinks of $\{e_1,\ldots,e_R\}$ as the columns of a matrix $M$, then the above matroid structure is known as the column matroid of $M$. 
    
We can now introduce bases of a matroid as 
\begin{definition}
Let $\mathcal{M}(X,\mathcal{I})$ be a finite matroid. A basis is a largest element $B\in \mathcal{I}$, that is for all $I\in \mathcal{I}, \lvert B\rvert\ge \lvert I\rvert.$
\end{definition}

It is straightforward to show that all bases of a matroid have the same cardinality.

We now define the dual of a matroid
\begin{definition}
The dual $\mathcal{M}^*(X,\mathcal{I}^*)$ of a finite matroid $\mathcal{M}(X,\mathcal{I})$ is the matroid over $X$ whose set of independents $\mathcal{I}^*$ is formed of subsets $I\subseteq X$ whose complement $\bar I$ contain a basis of $\mathcal{M}(X,\mathcal{I})$. 
\end{definition}

Alternatively, the dual of a finite matroid $\mathcal{M}(X,\mathcal{I})$ on $X$ is the matroid whose basis sets are the complements of the basis sets of $\mathcal{M}(X,\mathcal{I})$.

The rank function of a given matroid is defined on the subsets of $X$, 
\begin{eqnarray}
    \text{rk}:2^X&\rightarrow& \mathbb{Z}_+\\
    S&\mapsto &\max\{\lvert I \rvert: I\subseteq S, I\in \mathcal{I}\},
\end{eqnarray}
that is the size of the largest independent contained in $S$. One shows that the rank function of the dual matroid is \begin{equation}\label{eq:dual-rank}
    \text{rk}^*(S)=\text{rk}(X\backslash S)+\lvert S\rvert-\text{rk}(X)
\end{equation}
The cornerstone theorem of matroid theory is arguably Edmonds' intersection theorem \cite{edmonds2003submodular}:

\begin{theorem}[Edmonds' theorem - The matroid intersection theorem]\label{thm:Edmond}
Let $\mathcal{M}_1(X,\mathcal{I}_1), \mathcal{M}_2(X,\mathcal{I}_2)$ be two matroids over the same set $X$. The following equality holds
\begin{equation}
    \max_{I\in \mathcal{I}_1\cap \mathcal{I}_2} \lvert I \rvert=\min_{S\subseteq X}[\mathrm{rk}_1(S)+\mathrm{rk}_2(\bar S)],
\end{equation}
where the complement $\bar S$ of $S$ is meant in $X$, that is $\bar S= X\backslash S$.
\end{theorem}

\subsection{Proof of Theorem \ref{thm:pre_main}}
\begin{proof}
Let $\mathbf{G}$ be a generator matrix of the code $C$ of theorem \ref{thm:pre_main}. Let $j:=j(C)$ be as announced, \textit{i.e.} the smallest integer such that there exists a bipartition $A\sqcup \bar{A}=[n]$ inducing a submatrix $\mathbf{G}_A$ of rank $k$ and another submatrix $\mathbf{G}_{\bar{A}}$ of rank $k-j$.

Let $\mathcal{M}_1$ be the column matroid of $\mathbf{G}$, meaning the matroid on $X=[n]$ whose independent sets are the subsets $I\subset [n]$ such that the columns of $\mathbf{G}$ indexed by $I$ are linearly independent. We denote by $\text{rk}_1$ its rank function. We let $\mathcal{M}_2$ be the dual matroid of $\mathcal{M}_1$ and denote by $\text{rk}_2$ its rank function.

We remark that with those matroid structures, by definition of $j$, $k-j$ is the maximal value of $\text{rk}_1$ over the complements of bases of $\mathcal{M}_1$. Moreover, letting $B$ be a basis of $\mathcal{M}_1(X,\mathcal{I}_1)$, we have $\text{rk}_1(\bar B)=\max\{\lvert I\rvert:I\in \mathcal{I}_1, I\subseteq \bar B\}$. Hence, by definition of the dual matroid such $I$'s are also independent in $\mathcal{M}_2$, $k-j$ is the largest size of a set $I\subset [n]$ such that $I$ is independent in both $\mathcal{M}_1$ and $\mathcal{M}_2$. Therefore, letting $\mathcal{I}_1,\mathcal{I}_2$ be the sets of independents of respectively $\mathcal{M}_1,\mathcal{M}_2$, Edmonds' intersection theorem \cite{edmonds2003submodular} (Theorem \ref{thm:Edmond}) tells us that
\begin{equation}\label{eq:intersec-Edmond-proof}
    k-j = \max_{I\in \mathcal{I}_1\cap\mathcal{I}_2}\lvert I\rvert=\min_{S\subset [n]}[\text{rk}_1(S)+\text{rk}_2(\bar{S})].
\end{equation}
Consider $S\subset [n]$ achieving the minimum of the right-hand side of \eqref{eq:intersec-Edmond-proof}. Since we have supposed $j\geq 1$, there must exist $k_0>0$ such that $\text{rk}_1(S)=k-k_0$. 
We now let $x=\lvert \bar S\rvert$ so that, applying~\eqref{eq:dual-rank}, we have $\text{rk}_2(\bar S) =\mathrm{rk}_1(S)+|\bar{S}|-k=x-k_0$.

Since $S$ achieves the minimum of the right-hand side of \eqref{eq:intersec-Edmond-proof} we now have $k-j=(k-k_0)+(x-k_0)$ so that $\lvert \bar S\rvert=x=2k_0-j$. We now subdivide the matrix $\mathbf{G}$ into the set of columns indexed by $S$ and the set of columns indexed by $\bar S$. 
Assuming, without loss of generality, that the submatrix $\mathbf{G}_S$ whose columns are indexed by $S$ is the $k\times |S|$ leftmost submatrix of $\mathbf{G}$, we obtain, since $\mathrm{rk}(\mathbf{G}_S)=k-k_0$, that $\mathbf{G}$ multiplied on the left by an appropriate invertible matrix is of the form
\begin{equation}
\mathbf{G}'=
\left(
\begin{array}{c|c}
\mathbf{0} & \mathbf{G}_0 \\ \hline
\mathbf{G}_1 & \mathbf{G}_2
\end{array}
\right)
\end{equation}
where the top left $k_0\times |S|$ submatrix is the zero matrix. The top right submatrix $\mathbf{G}_0$ must therefore be
a $k_0\times (2k_0-j)$ matrix, and it must be of rank $k_0$ for the rank of $\mathbf{G}'$ to be of rank $k$. Since $\mathbf{G}'$ is, like $\mathbf{G}$, a generator matrix for the code $C$, we have that the shortened code $C_0$ of $C$ supported by $\bar{S}$ has dimension $k_0$ and length $2k_0-j$.
\end{proof}

\bibliographystyle{alpha}
\bibliography{biblio_QEC_norm}
\end{document}